\documentclass[11pt,a4paper]{article}
\usepackage[T1]{fontenc}
\usepackage{lmodern}
\usepackage[a4paper,margin=1in,includefoot,footskip=24pt]{geometry}
\usepackage{setspace}
\usepackage{amsmath,amssymb,amsthm}
\usepackage[hidelinks]{hyperref}
\usepackage[protrusion=false]{microtype}
\hypersetup{pdftitle={Complete bipartite embeddings, lattice non-embeddability, and non-Euclidean powers of information distance},pdfauthor={Fengqi Hou}}
\newtheorem{theorem}{Theorem}
\newtheorem{lemma}[theorem]{Lemma}
\newtheorem{corollary}[theorem]{Corollary}
\theoremstyle{remark}
\newtheorem{remark}[theorem]{Remark}
\newcommand{\bits}{\{0,1\}}
\newcommand{\F}{\mathbb F}
\newcommand{\N}{\mathbb N}
\newcommand{\R}{\mathbb R}
\newcommand{\dK}{d_K}
\newcommand{\pref}[2]{#1\mathbin{\upharpoonright}#2}
\newcommand{\cat}{\mathbin{\Vert}}
\newcommand{\ceil}[1]{\left\lceil #1\right\rceil}

\DeclareMathOperator{\bin}{bin}
\title{Complete bipartite embeddings, lattice non-embeddability, and non-Euclidean powers of information distance}
\author{Fengqi Hou\\[0.4em]
\normalsize Beijing Innovation Center of Humanoid Robotics (X-Humanoid)}
\date{}
\begin{document}
\raggedbottom
\maketitle
\begin{abstract}
We encode the vertices of every finite complete bipartite graph as binary strings at each scale. The information distances between these strings approximate the scaled graph distances with an additive error bounded by a constant independent of the scale. We also assign one fixed infinite binary sequence to each vertex; suitable prefixes satisfy the same distance estimate with logarithmic error. This extends Hutter's construction for $K_{3,3}$ and resolves his open problem on scale-embeddings of finite complete bipartite graphs. Combined with his complete bipartite obstruction, these embeddings show that no positive power of information distance can be represented exactly by distances in a real Hilbert space, resolving his question about such representations. For conditional prefix complexity, we further prove that at most $2^{\delta+C}$ strings lie between any two strings with total distance exceeding their mutual distance by at most $\delta$, where $C$ is independent of the endpoints. This bound rules out scale-embeddings of the integer line and every positive-dimensional integer lattice, resolving Hutter's lattice embedding problem for both string and sequence embeddings. Restriction to the integer lattice gives the same conclusion for $(\mathbb R^m,\|\cdot\|_1)$ for every $m\ge1$.
\end{abstract}
\medskip

\section{Introduction}
Information distance measures the information needed to transform one object into another. Bennett, G\'acs, Li, Vit\'anyi and Zurek~\cite{bennett} established its basic properties. We use the maximum of the two conditional prefix complexities, set the distance from a string to itself to zero, and add a fixed constant to the other distances as in Hutter~\cite{hutter}.

Hutter introduced scale-embeddings into information distance, allowing logarithmic additive error~\cite[Definitions 8--9]{hutter}. In a string embedding, each point is encoded as a finite binary string at each scale. In a sequence embedding, each point is assigned one infinite binary sequence. As the scale changes, only the length of the prefix used may change; the sequence itself remains the same. Hutter constructed such an embedding for $K_{3,3}$~\cite[proof of Proposition 26]{hutter}. Open Problem 37 asks whether every finite complete bipartite graph admits a scale-embedding; we establish both the string and sequence versions.

We extend this construction to every finite complete bipartite graph. For finite strings, we encode the vertices so that the information distance between their encodings differs from $s$ times their graph distance by at most a constant independent of $s$. Thus, at scale $s$, the distances are approximately $s$ between opposite sides and $2s$ between distinct vertices on the same side. We also assign each vertex a fixed infinite binary sequence. At each scale $s$, we choose prefixes so that the information distance between them differs from $s$ times the corresponding graph distance by at most $A\log_2(s+2)+C$, where $A$ and $C$ are independent of the scale and the vertex pair. These are the two parts of Theorem~\ref{thm:main}; the second answers Open Problem 37. The construction uses $2\times2$ matrices over a sufficiently large finite field. Its algebraic and coding details are given in Section~\ref{sec:bipartite}.

Hutter's non-Euclidean criterion for powers of complete bipartite graph distances~\cite[Theorem 36]{hutter} connects this result to Open Problem 39. As he observed, an affirmative answer to Open Problem 37 would imply that every positive power of information distance is non-Euclidean. The $K_{3,3}$ construction already gives this conclusion for
\begin{equation}\label{eq:oldthreshold}
 \alpha>\alpha_0:=\log_4(3/2)=0.29248125036\ldots;
\end{equation}
the proof of~\cite[Theorem 33]{hutter} gives this threshold, while the theorem is stated with the simpler bound $\alpha>0.3$. Our embeddings cover the remaining range $0<\alpha\le\alpha_0$ (Corollary~\ref{cor:allpowers}). For each positive exponent, prefixes of a fixed finite family of infinite sequences also give non-Euclidean finite sets at every sufficiently large selected scale (Corollary~\ref{cor:prefix}).

Hutter's Corollary 24 implies that every bounded subset of the integer lattice $(\mathbb Z^m,\|\cdot\|_1)$ admits a scale-embedding into $B_K^*$. Open Problem 29 asks whether the entire lattice admits such an embedding. We answer this negatively for conditional prefix complexity. For a fixed allowed excess $\delta$, we bound the number of strings that lie between two given strings in the following sense: the sum of the two distances through an intermediate string exceeds the distance between the given strings by at most $\delta$. The bound does not depend on those two strings (Lemma~\ref{lem:interval}). It implies that the integer line $\mathbb Z$ cannot scale-embed into $B_K^*$. An embedding of $\mathbb Z^m$ would restrict to such an embedding on a coordinate axis, so no positive-dimensional integer lattice admits one (Corollary~\ref{cor:lattice}). Restriction to the integer lattice also rules out the whole space $(\mathbb R^m,\|\cdot\|_1)$ (Corollary~\ref{cor:realspace}). The optional monotone-complexity convention in Hutter's Assumption 10 is discussed in Remark~\ref{rem:lattice-scope}.

Section~\ref{sec:main} states the main results. Section~\ref{sec:bipartite} constructs the bipartite embeddings, and Section~\ref{sec:powers} proves the consequences for positive powers. Section~\ref{sec:intervals} gives the independent interval-counting argument for integer lattices. Section~\ref{sec:conclusion} concludes the paper. The appendix describes the scope of the Lean formalization. All problem and theorem numbers referring to~\cite{hutter} are those of arXiv:2507.21988v1.

\section{Definitions and main results}\label{sec:main}
\subsection{Information distance and scale-embeddings}
Fix an optimal conditional prefix machine $U$, and put
\[
 K(x\mid y)=\min\{|p|:U(p,y)=x\}.
\]
The machine takes the program $p$ and the condition $y$ as separate inputs. For each fixed $y$, the programs on which it halts form a prefix-free set: no such program is a proper prefix of another.

Following~\cite[Definition 4 and Theorem 5]{hutter}, fix a sufficiently large nonnegative constant $c_{\rm ref}$ and define
\begin{equation}\label{eq:distance}
 \dK(x,y)=
 \begin{cases}
 \max\{K(x\mid y),K(y\mid x)\}+c_{\rm ref},&x\ne y,\\
 0,&x=y.
 \end{cases}
\end{equation}
We choose $c_{\rm ref}$ large enough to make $\dK$ a metric. In the rest of this paper, the constants in our error bounds include $c_{\mathrm{ref}}$. We write $B_K^*=(\bits^*,\dK)$.

Write $\pref{\xi}{t}$ for the prefix of length $t$ of an infinite binary sequence $\xi$. For a metric space $(S,d)$, a string scale-embedding is a family of maps $\phi_s:S\to\bits^*$ whose distances satisfy
\[
 \left|\dK(\phi_s(a),\phi_s(b))-s\,d(a,b)\right|
 \le A\log_2(s+2)+C
\]
for all $a,b\in S$ and all positive integer scales $s$, with constants independent of $s,a,b$. A sequence scale-embedding uses maps of the form $\phi_s(a)=\pref{\xi_a}{L_a(s)}$, where $\xi_a$ is fixed for each point and $L_a(s)$ is nondecreasing in $s$. In the prefix-complexity convention used here, Hutter's Assumption 10 also requires at most polynomial growth of the selected lengths~\cite[Definitions 8--9 and Assumption 10]{hutter}. Replacing $\log_2(s+2)$ by $\log_2 s$ changes only the additive constant for $s\ge1$. Our positive constructions have a common length $L(s)=2s+O(1)$. Our nonexistence results remain valid if the distance estimate is required only at sufficiently large scales and the prefix lengths are arbitrary functions of the scale and the point.

\subsection{Complete bipartite graph embeddings}

Let $m,n$ be positive integers and let $G=K_{m,n}$ have vertices
\[
 V(G)=\{L_1,\ldots,L_m\}\sqcup\{R_1,\ldots,R_n\}.
\]
Its shortest-path distance $d_G$ is zero on the diagonal, one between the sides, and two between distinct vertices on the same side. Both sides are required to be nonempty. Graph vertices are denoted by $a,b$; the direction vectors $u_i,v_j$ used in the construction will be defined in Section~\ref{sec:bipartite}.

\begin{theorem}\label{thm:main}
For every pair of positive integers $m,n$, there are integers $r,B\ge1$ and constants $C_{\rm fin},C_{\rm seq}\ge0$ with the following properties. Set
\begin{equation}\label{eq:length}
 L(s)=B+2r\ceil{s/r}\qquad(s\in\N,\ s\ge1).
\end{equation}
\begin{enumerate}
\item For every integer $s\ge1$, there is an injection $\phi_s:V(G)\to\bits^{L(s)}$ such that
\begin{equation}\label{eq:finite}
 \left|\dK(\phi_s(a),\phi_s(b))-s d_G(a,b)\right|\le C_{\rm fin}
\end{equation}
for all vertices $a,b$.
\item There are fixed infinite binary sequences $(\xi_a)_{a\in V(G)}$ such that $\psi_s(a)=\pref{\xi_a}{L(s)}$ is injective for every integer $s\ge1$ and
\begin{equation}\label{eq:nested}
 \left|\dK(\psi_s(a),\psi_s(b))-s d_G(a,b)\right|
 \le 2\log_2(s+2)+C_{\rm seq}
\end{equation}
for all vertices $a,b$ and all such $s$.
\end{enumerate}
The common length is nondecreasing, tends to infinity, and satisfies
\begin{equation}\label{eq:lengthbounds}
 B+2s\le L(s)<B+2s+2r.
\end{equation}
\end{theorem}

Theorem~\ref{thm:main} answers Hutter's Open Problem 37 affirmatively, providing both string and sequence scale-embeddings.

The constants may depend on the fixed graph, the reference machine and the algebraic and coding choices, as well as the chosen infinite source in the sequence construction. They are independent of the scale and the vertex pair. Both parts of Theorem~\ref{thm:main} assert existence: no algorithm for selecting the finite data at each scale or generating the infinite source is required. The first part satisfies a constant error bound, which is stronger than the logarithmic bound allowed in~\cite[Definition 8]{hutter}. The logarithmic error and prefix lengths in the second part satisfy~\cite[Definition 9 and Assumption 10]{hutter}.

\subsection{Positive powers of information distance and Hilbert-space representations}
The embeddings above, combined with Hutter's criterion, give the following two consequences. The proofs appear in Section~\ref{sec:powers}.

\begin{corollary}\label{cor:allpowers}
For every real $\alpha>0$, there is no map $F$ from the set $\bits^*$ of all finite binary strings into any real Hilbert space $H$ such that
\begin{equation}\label{eq:representation}
 \|F(x)-F(y)\|_H=\dK(x,y)^\alpha
 \qquad\text{for all }x,y\in\bits^*.
\end{equation}
Here $\|\cdot\|_H$ is the norm induced by the inner product on $H$. Thus Open Problem 39 of~\cite{hutter} has a negative answer.
\end{corollary}

\begin{corollary}\label{cor:prefix}
For each real $\alpha>0$, there are an integer $N$, a fixed family of infinite binary sequences $(\xi_a)_{a\in V(K_{N,N})}$, a common length function $L(s)=B+2r\ceil{s/r}=2s+O(1)$, and an integer $s_0\ge1$ such that, for every integer $s\ge s_0$, the finite set
\[
 X_s=\{\pref{\xi_a}{L(s)}:a\in V(K_{N,N})\}
\]
with pairwise distances $\dK^\alpha$ has no isometric representation in a real Hilbert space.
\end{corollary}

\subsection{Integer lattices}
We next show that $(\mathbb Z^m,\|\cdot\|_1)$ cannot scale-embed into $B_K^*$ for any $m\ge1$. The proof is given in Section~\ref{sec:intervals}.

\begin{corollary}\label{cor:lattice}
For every integer $m\ge1$, $(\mathbb Z^m,\|\cdot\|_1)$ has neither a string scale-embedding nor a sequence scale-embedding into the distance~\eqref{eq:distance} under Definitions 8--9 of~\cite{hutter}. Thus Open Problem 29 has a negative answer for the prefix-complexity distance.
\end{corollary}

\section{Complete bipartite graph embeddings}\label{sec:bipartite}
We prove Theorem~\ref{thm:main} in four steps, corresponding to the four subsections below. First, we use matrix products to construct a binary string for each vertex. Second, we bound the length of a shortest prefix program that outputs one vertex's string when given another vertex's string. Third, we choose a suitable finite sequence of matrices at each scale and add fixed vertex labels to obtain the finite-string embedding. Fourth, we fix a suitable infinite sequence of matrices and prove that the resulting vertex sequences, with the same fixed labels, satisfy the required distance estimates for their prefixes.

\subsection{The finite-field matrix construction}
Choose $r\ge1$ such that $q=2^r$ satisfies $q+1\ge\max(m,n)$, and fix the field $\F_q$. Since $\F_q$ has $2^r$ elements, each element can be encoded by a distinct binary string of length $r$. We fix a basis of $\F_q$ over $\F_2$ and use the $r$ binary coordinates in this basis as the encoding.

Assign each left vertex $L_i$ a nonzero column vector $u_i\in\F_q^2$, and each right vertex $R_j$ a nonzero column vector $v_j\in\F_q^2$. We regard nonzero scalar multiples as representing the same direction. On each side, no two chosen vectors are scalar multiples of one another. Vectors on opposite sides may coincide. There are $q+1$ available directions, represented by
\[
 (1,\lambda)^T\quad(\lambda\in\F_q),\qquad (0,1)^T,
\]
so the chosen field has enough directions for both sides. These direction vectors are fixed parameters in the matrix multiplications and remain unchanged as the scale changes.

A matrix $M\in\F_q^{2\times2}$ has four field elements, which we encode in row order as a $4r$-bit string. For each vertex, define the corresponding matrix product by
\begin{equation}\label{eq:observations}
 X_i(M)=u_i^TM,\qquad Y_j(M)=Mv_j.
\end{equation}
Here $u_i^TM$ is a row vector and $Mv_j$ is a column vector. Each product has two field elements and is therefore encoded as a $2r$-bit string. Fix a finite sequence of matrices $M_1,\ldots,M_\ell$. For each vertex, we encode the products obtained from $M_1,\ldots,M_\ell$ as binary strings and concatenate them in that order. The following lemma gives the relations between the products that we will use to bound the conditional complexity of these strings.

\begin{lemma}\label{lem:recovery}
For the fixed direction vectors chosen above, two products $u_i^TM$ and $u_j^TM$ with $i\ne j$ determine $M$; the same holds for $Mv_i$ and $Mv_j$. For vertices on opposite sides, either product $u_i^TM$ or $Mv_j$ can be computed from the other using one suitable additional field element. Given both products, one further suitable field element suffices to recover $M$.
\end{lemma}
\begin{proof}
For two distinct vertices on the left side, their direction vectors $u_i$ and $u_j$ are not proportional. Hence the matrix with rows $u_i^T,u_j^T$ is invertible. Given the two products $X_i(M)$ and $X_j(M)$, we recover $M$ using the fixed direction vectors $u_i,u_j$:
\[
 M=\begin{pmatrix}u_i^T\\u_j^T\end{pmatrix}^{-1}
 \begin{pmatrix}X_i(M)\\X_j(M)\end{pmatrix}
 .
\]
The same argument on the right uses the invertible matrix with columns $v_i,v_j$.

Fix opposite-side directions $u,v$. Choose a vector $u'$ that is not proportional to $u$, and a vector $v'$ that is not proportional to $v$. Then $u,u'$ and $v,v'$ each form a basis of $\F_q^2$. Put
\[
 P=\begin{pmatrix}u^T\\(u')^T\end{pmatrix},\qquad
 Q=\begin{pmatrix}v&v'\end{pmatrix},\qquad \widetilde M=PMQ.
\]
Writing $X=u^TM$ and $Y=Mv$, we have
\[
 XQ=(\widetilde M_{11},\widetilde M_{12}),\qquad
 PY=(\widetilde M_{11},\widetilde M_{21})^T.
\]
Given $Y$, compute $PY$ and hence $\widetilde M_{11}$. Supplying the value of $\widetilde M_{12}$ then gives $X=(\widetilde M_{11},\widetilde M_{12})Q^{-1}$. In the other direction, $X$ determines $\widetilde M_{11}$ and the supplied value of $\widetilde M_{21}$ gives $Y=P^{-1}(\widetilde M_{11},\widetilde M_{21})^T$. Thus only one additional field element is supplied in either direction. Given both $X$ and $Y$, and the value of $\widetilde M_{22}$, all four entries of $\widetilde M$ are known, and $M=P^{-1}\widetilde M Q^{-1}$. The directions and the matrices $P,Q$ are fixed for this vertex pair; we can compute $X$ from $Y$ and the supplied value of $\widetilde M_{12}$, and compute $Y$ from $X$ and the supplied value of $\widetilde M_{21}$. Neither conversion requires the full matrix $M$. This calculation is valid also when $u^Tv=0$.
\end{proof}

For $\ell\ge1$, let $W_\ell$ be the concatenation of the $4r$-bit encodings of $M_1,\ldots,M_\ell$. At a left vertex $L_i$, concatenate the $2r$-bit encodings of $u_i^TM_1,\ldots,u_i^TM_\ell$ to form $x_{L_i}$. At a right vertex $R_j$, use $M_1v_j,\ldots,M_\ell v_j$ in the same way to form $x_{R_j}$. Thus $W_\ell$ encodes all the matrices, while $x_a$ is the string constructed for vertex $a$, before adding its fixed label. With $T=r\ell$, we have
\begin{equation}\label{eq:rawlength}
 |W_\ell|=4T,\qquad |x_a|=2T.
\end{equation}
All transformations are computable once the finite field, directions and encoding have been fixed. As the scale increases, we use more matrices; each matrix remains $2\times2$ and each direction vector has two components.

\begin{samepage}
\subsection{Conditional complexity estimates}

\begin{lemma}\label{lem:complexity}
There is a constant $C_0\ge0$ with the following property. Let $\ell\ge1$ and $\Delta\ge0$, and let $W_\ell$ be a matrix encoding satisfying
\begin{equation}\label{eq:deficiency}
 K(W_\ell\mid\ell)\ge4r\ell-\Delta.
\end{equation}
Set $T=r\ell$. Then, for any two distinct vertices $a,b$, the corresponding strings $x_a,x_b$ satisfy
\begin{equation}\label{eq:complexity}
 T d_G(a,b)-\Delta-C_0
 \le K(x_a\mid x_b)\le T d_G(a,b)+C_0.
\end{equation}
\end{lemma}
\end{samepage}
Section~\ref{sec:finite} proves that, for each $\ell$, there are matrix data satisfying~\eqref{eq:deficiency} with $\Delta=0$. Section~\ref{sec:infinite} proves that the prefixes of one fixed infinite sequence satisfy~\eqref{eq:deficiency} with $\Delta=\Delta_\ell=O(\log(\ell+2))$.

\begin{proof}
To prove the upper bounds, we use a program that takes $x_b$ as a condition and reads additional binary data to produce $x_a$. For vertices on the same side, we supply the entire $2T$-bit string $x_a$; for vertices on opposite sides, we supply one field element per matrix, encoded in $T$ bits in total. Since $T=|x_b|/2$, the program knows the required data length in each case without a separate description of $\ell$.

In each case, for a fixed $x_b$, all valid additional binary data have the same length and therefore form a prefix code. By the optimality of $U$, there is a program $p$ that outputs $x_a$ given $x_b$, whose length is at most the length of the additional binary data plus a fixed constant.

For distinct vertices on the same side, copying the target string gives
\[
 K(x_a\mid x_b)\le2T+O(1),
\]
and the same bound holds with $a$ and $b$ interchanged.

Let $k=K(x_a\mid x_b)$, and let $p$ be a shortest prefix program that outputs $x_a$ given $x_b$, so that $|p|=k$.

We construct a program that outputs $W_\ell$ with $\ell$ as its only condition. The program stores the full $2T$-bit string $x_b$, followed by the $k$-bit program $p$. Since $r$ is fixed, it computes $2T=2r\ell$ from $\ell$ and reads the first $2T$ bits of the stored data to obtain $x_b$. It then runs $p$ with $x_b$ as the condition, obtaining $x_a$. The two strings provide the products at two distinct vertices on the same side for every matrix. Lemma~\ref{lem:recovery} therefore allows the program to reconstruct all the matrices and output their concatenated encoding $W_\ell$.

The code for carrying out these steps adds a fixed constant $C$, giving
\[
 K(W_\ell\mid\ell)\le2T+k+C.
\]
The $2T$ bits of $x_b$ are included in the program length because only $\ell$ is given as a condition. The string $x_a$ need not be stored separately: it is generated by $p$.

By assumption~\eqref{eq:deficiency}, $4T-\Delta\le K(W_\ell\mid\ell)$. Combining this with the preceding upper bound gives $k\ge2T-\Delta-C$.

For a pair of vertices on opposite sides, apply the transformation $\widetilde M_t=PM_tQ$ from Lemma~\ref{lem:recovery} to each matrix $M_1,\ldots,M_\ell$, using the same $P,Q$ for all matrices. Given the string for a right vertex, decode its products $M_1v_j,\ldots,M_\ell v_j$. The list of $(1,2)$ entries of the transformed matrices, encoded in $T$ bits, then lets us compute the products $u_i^TM_1,\ldots,u_i^TM_\ell$ and hence the string for the left vertex. In the other direction the list of $(2,1)$ entries of the transformed matrices suffices. Thus
\[
 K(x_a\mid x_b)\le T+O(1),
\]
and the same bound holds with $a$ and $b$ interchanged.

For the reverse inequality, describe $W_\ell$ conditional on $\ell$ by the $2T$ bits of $x_b$, a shortest conditional prefix program of length $k=K(x_a\mid x_b)$ for $x_a$, and the $T$ bits for the $(2,2)$ entries of the transformed matrices. Decoding $x_a,x_b$ gives the two products for each matrix; these products and the supplied $(2,2)$ entries recover all the matrices, so
\[
 4T-\Delta\le K(W_\ell\mid\ell)\le3T+k+O(1).
\]
Hence $k\ge T-\Delta-O(1)$.

Each recovery procedure receives the concatenation of $x_b$, $p$, and the additional data as a complete finite binary string, with $\ell$ given separately as the condition. The total input length is therefore available. Since $\ell$ determines $T=r\ell$, the procedure takes the first $2T$ bits as $x_b$. In the opposite-side case, it also takes the last $T$ bits as the additional data; in the same-side case, there are no additional data. The remaining bits form $p$, which is run on $U$ with $x_b$ as the condition to obtain $x_a$.

For each recovery procedure and fixed $\ell$, the data before and after $p$ have fixed lengths. For each fixed $x_b$, the programs $p$ on which $U(p,x_b)$ halts form a prefix-free set. Hence the complete inputs on which the recovery procedure halts also form a prefix-free set.

Thus a single constant $C_0$ works for all distinct vertices $a,b$, all lengths $\ell$, and all matrix data $W_\ell$ satisfying~\eqref{eq:deficiency}.
\end{proof}

\subsection{The finite-string embedding}\label{sec:finite}
\begin{proof}[Proof of Theorem~\ref{thm:main}(1)]
We now choose matrix data satisfying the hypothesis of Lemma~\ref{lem:complexity} with $\Delta=0$.
For every fixed $\ell$, there are $2^{4T}$ strings of length $4T$ and fewer than $2^{4T}$ programs of length below $4T$. Consequently, some $W_\ell$ satisfies $K(W_\ell\mid\ell)\ge4T$. Lemma~\ref{lem:complexity}, with $\Delta=0$, then gives, simultaneously for all distinct vertices,
\begin{equation}\label{eq:finitecomplexity}
 K(x_a\mid x_b)=T d_G(a,b)+O(1).
\end{equation}

To ensure injectivity at every scale, we prepend a fixed binary label, or tag, to each vertex's string. Write $\cat$ for string concatenation. Put $k_0=\ceil{\log_2\max(m,n)}$ and $B=1+k_0$, and assign the tags
\[
 \tau(L_i)=0\cat\bin_{k_0}(i-1),\qquad
 \tau(R_j)=1\cat\bin_{k_0}(j-1).
\]
Here $\bin_{k_0}$ denotes a fixed-length binary encoding; when $k_0=0$, the index string is empty. The tags have common length $B$ and are pairwise distinct. For every fixed pair of vertices,
\begin{equation}\label{eq:tags}
 K(\tau(a)\cat x_a\mid\tau(b)\cat x_b)
 =K(x_a\mid x_b)+O(1).
\end{equation}
Adding or removing the fixed tags from the condition and the output costs $O(1)$. There are only finitely many vertex pairs, so the constant is uniform.

For an integer $s\ge1$, take $\ell=\ceil{s/r}$, choose the incompressible $W_\ell$ at that length, and set $\phi_s(a)=\tau(a)\cat x_a$. This is injective and has length $B+2r\ceil{s/r}$. Since
\begin{equation}\label{eq:rounding}
 s\le T=r\ceil{s/r}<s+r
\end{equation}
and $d_G\le2$, equations~\eqref{eq:distance}, \eqref{eq:finitecomplexity} and~\eqref{eq:tags} give~\eqref{eq:finite}, with a constant independent of $s$. The diagonal case is exact. Equation~\eqref{eq:rounding} also gives~\eqref{eq:lengthbounds}.
\end{proof}

\subsection{The fixed infinite sequences}\label{sec:infinite}
Choose and fix an infinite binary sequence $\omega$ that is Martin--L\"of random for the fair-coin measure. Such sequences exist. By the Levin--Schnorr theorem~\cite[Theorem 92, p.~162]{shen}, there is a constant $c$ such that
\[
 K(\pref{\omega}{n})\ge n-c\qquad(n\ge1).
\]
Here $K(x)$ denotes unconditional prefix complexity. The standard inequality $K(x)\le K(n)+K(x\mid n)+O(1)$ follows from~\cite[Theorem 66, p.~119]{shen} by taking the pair $(n,x)$ and projecting onto $x$. Together with $K(n)=O(\log_2(n+2))$~\cite[Section 4.6, p.~110]{shen}, it gives
\[
 K(\pref{\omega}{n}\mid n)\ge n-O(\log_2(n+2)).
\]
If $\omega$ were computable, a fixed program given $n$ would output $\pref{\omega}{n}$, yielding $K(\pref{\omega}{n}\mid n)=O(1)$. This contradicts the preceding lower bound, so $\omega$ is not computable. The construction below proves existence without providing an algorithm that generates this source; it uses the complexity lower bound above.

\begin{proof}[Proof of Theorem~\ref{thm:main}(2)]
For every $\ell\ge1$, put $W_\ell=\pref{\omega}{4r\ell}$. Then $K(W_\ell)\ge4r\ell-c$. The same complexity inequality, now with condition $\ell$, gives
\[
 K(W_\ell)\le K(\ell)+K(W_\ell\mid\ell)+O(1).
\]
Using the bound $K(\ell)\le2\log_2(\ell+2)+O(1)$~\cite[Section 4.6, p.~110]{shen}, choose a constant $D\ge0$, independent of $\ell$, so that
\begin{equation}\label{eq:allprefix}
 K(W_\ell\mid\ell)\ge4r\ell-\Delta_\ell,
 \qquad \Delta_\ell:=2\log_2(\ell+2)+D.
\end{equation}

Divide $\omega$ into successive groups of $4r$ bits and decode each group as a matrix in $\F_q^{2\times2}$. At each vertex, compute its product in~\eqref{eq:observations} for each matrix, encode the product as a $2r$-bit string, and concatenate these strings in order. Prepend the fixed tag $\tau(a)$ once, at the start, to obtain the infinite sequence $\xi_a$. Its prefix of length $B+2r\ell$ is exactly $\tau(a)\cat x_a$, where $x_a$ is formed from the first $\ell$ matrices encoded by $W_\ell$.

Apply Lemma~\ref{lem:complexity} to~\eqref{eq:allprefix}, and then use~\eqref{eq:tags} and~\eqref{eq:distance}. For a constant $C_1$ independent of $\ell$,
\begin{equation}\label{eq:prefixestimate}
 \left|\dK(\pref{\xi_a}{B+2r\ell},\pref{\xi_b}{B+2r\ell})
       -r\ell d_G(a,b)\right|
 \le\Delta_\ell+C_1
\end{equation}
for all pairs, with the diagonal again exact. Take $\ell=\ceil{s/r}$. Since $s$ is a positive integer and $r\ge1$, we have $\ell\le s$ and
\[
 \Delta_\ell\le2\log_2(s+2)+D.
\]
The rounding error contributes at most $2r$ since $d_G\le2$, proving~\eqref{eq:nested}. The sequences are fixed once $\omega$ is chosen; increasing $s$ only lengthens their prefixes. The tags make these prefixes injective at every selected scale. The formula for $L(s)$ shows that the prefix length is nondecreasing in $s$. Equation~\eqref{eq:lengthbounds} gives $L(s)=2s+O(1)$, so the prefix lengths grow linearly with $s$ and tend to infinity.
\end{proof}

\section{Non-Euclidean powers of information distance}\label{sec:powers}
We prove that no positive power of information distance can be represented exactly by distances in a real Hilbert space, giving a negative answer to Hutter's Open Problem 39.

The following standard inner-product identity gives the squared-distance condition in~\cite[Section 3, equations (5)--(6), p.~526]{schoenberg}. The calculation for $K_{N,N}$ below is the special case $m=n=N$ of~\cite[Theorem 36]{hutter}.

\begin{lemma}\label{lem:hilbert}
If $z_1,\ldots,z_t$ belong to a real Hilbert space $H$ and $\lambda_1,\ldots,\lambda_t\in\R$ satisfy $\sum_i\lambda_i=0$, then
\begin{equation}\label{eq:hilbert}
 \sum_{i,j}\lambda_i\lambda_j\|z_i-z_j\|_H^2
 =-2\left\|\sum_i\lambda_i z_i\right\|_H^2\le0.
\end{equation}
\end{lemma}
Thus a set of proposed distances cannot be represented in a real Hilbert space if some real weights summing to zero make the weighted sum of squared distances positive.

\begin{samepage}
\begin{lemma}[Hutter's balanced bipartite obstruction]\label{lem:obstruction}
Let $N$ be a positive integer. For $G=K_{N,N}$, assign weight $1/N$ to each left vertex and $-1/N$ to each right vertex. These real weights satisfy $\sum_{a\in V(G)}\lambda_a=0$. In the following sum, $a,b$ range independently over all vertices on both sides. For every $\alpha>0$,
\begin{equation}\label{eq:obstruction}
 S_N(\alpha):=\sum_{a,b}\lambda_a\lambda_b d_G(a,b)^{2\alpha}
 =2(1-1/N)4^\alpha-2.
\end{equation}
In particular, $S_N(\alpha)>0$ whenever
\begin{equation}\label{eq:Nchoice}
 N>\frac{1}{1-4^{-\alpha}}.
\end{equation}
\end{lemma}
\end{samepage}
\begin{proof}
For $N\ge2$, setting $m=n=N$ in the calculation in~\cite[Theorem 36]{hutter} gives~\eqref{eq:obstruction}; for $N=1$, both sides equal $-2$. Solving $S_N(\alpha)>0$ for $N$ gives~\eqref{eq:Nchoice}.
\end{proof}

\begin{proof}[Proof of Corollary~\ref{cor:allpowers}]
Fix $\alpha>0$ and then an integer $N$ satisfying~\eqref{eq:Nchoice}. Apply Theorem~\ref{thm:main}(1) to $K_{N,N}$. Its embeddings satisfy
\[
 \frac{\dK(\phi_s(a),\phi_s(b))}{s}=d_G(a,b)+O(1/s).
\]
For the zero-sum weights of Lemma~\ref{lem:obstruction}, put
\[
 Q_s=\sum_{a,b}\lambda_a\lambda_b
 \left(\frac{\dK(\phi_s(a),\phi_s(b))}{s}\right)^{2\alpha}.
\]
Continuity of the positive power function on the nonnegative reals and finiteness of the sum give
\begin{equation}\label{eq:positive-limit}
 Q_s\longrightarrow S_N(\alpha)>0.
\end{equation}
Suppose that $F:\bits^*\to H$ satisfied~\eqref{eq:representation}. Here $a$ is a graph vertex, $\phi_s(a)$ is its binary string, and $F(\phi_s(a))$ is a vector in $H$. Define the rescaled vectors
\[
 z_{s,a}=\frac{F(\phi_s(a))}{s^\alpha}.
\]
This divides the Hilbert-space vectors by the positive real number $s^\alpha$. Their squared distances are
\[
 \|z_{s,a}-z_{s,b}\|_H^2
 =\frac{\dK(\phi_s(a),\phi_s(b))^{2\alpha}}{s^{2\alpha}}
 =\left(\frac{\dK(\phi_s(a),\phi_s(b))}{s}\right)^{2\alpha}.
\]
The weights are $\lambda_a=1/N$ for a left vertex and $\lambda_a=-1/N$ for a right vertex; the symbol $\lambda_b$ denotes the weight of the second vertex in the same family. Lemma~\ref{lem:hilbert} therefore gives
\[
 Q_s=\sum_{a,b}\lambda_a\lambda_b\|z_{s,a}-z_{s,b}\|_H^2
 =-2\left\|\sum_a\lambda_a z_{s,a}\right\|_H^2\le0.
\]
This contradicts~\eqref{eq:positive-limit}.
\end{proof}

For each fixed $\alpha>0$, we choose $N$ and keep it fixed as $s\to\infty$. The embedding constant $C_{\rm fin}$ may depend on $N$, but does not depend on $s$, so $C_{\rm fin}/s\to0$. For $\alpha>1$, the conclusion concerns the representation~\eqref{eq:representation}, whether or not $\dK^\alpha$ satisfies the triangle inequality. We follow Hutter's use of ``Euclidean'' to include real Hilbert spaces.

\begin{proof}[Proof of Corollary~\ref{cor:prefix}]
Fix $\alpha$ and choose $N$ as in~\eqref{eq:Nchoice}. Use the fixed sequences supplied by Theorem~\ref{thm:main}(2). Their normalized distances satisfy
\[
 \frac{\dK(\psi_s(a),\psi_s(b))}{s}
 =d_G(a,b)+O\left(\frac{\log_2(s+2)}{s}\right)
 \longrightarrow d_G(a,b).
\]
The corresponding forms $Q_s$ again converge to $S_N(\alpha)>0$. For every sufficiently large integer $s$, we have $Q_s>0$. If the prefixes at that scale could be represented by points in a real Hilbert space with pairwise distances $\dK^\alpha$, Lemma~\ref{lem:hilbert} would give $Q_s\le0$, a contradiction. Therefore, no such representation exists.
\end{proof}

For each fixed $\alpha>0$, we choose an integer $N$ and an infinite binary sequence $\xi_a$ for each vertex $a$ of $K_{N,N}$. As the scale $s$ varies, $N$ and these sequences remain fixed; only the length $L(s)$ of the selected prefixes changes.

At each scale $s$, we allow a proposed representation to use its own real Hilbert space and its own map from the prefixes to points in that space. The conclusion rules out every such representation with pairwise distances $\dK^\alpha$ for all sufficiently large $s$. The conclusion applies to the selected lengths $L(s)$.

Hutter's $K_{3,3}$ sequence construction gives this conclusion for $\alpha>\alpha_0$. Our construction also covers the remaining range $0<\alpha\le\alpha_0$, and thus gives the conclusion for every $\alpha>0$.

\section{Non-embeddability of integer lattices}\label{sec:intervals}
The constants in this section depend only on $U$, the fixed rule for encoding ordered pairs of strings, and $c_{\rm ref}$. They do not depend on the endpoint strings $x,y$. Fix a computable pairing $\langle x,y\rangle$ with computable projections, and write $K(x,y\mid t)$ for $K(\langle x,y\rangle\mid t)$.

\subsection{Two prefix-complexity estimates}
The conditional chain rule for prefix complexity~\cite[Section 4.7.3, Problem 112, p.~123]{shen} gives
\begin{equation}\label{eq:precise-chain}
 K(x,y\mid t)=K(x\mid t)+K(y\mid t,x,K(x\mid t))+O(1).
\end{equation}
The integer $K(x\mid t)$ is supplied as part of the auxiliary input; it is needed for the $O(1)$ remainder. Exchanging the entries of a pair also changes complexity by at most $O(1)$.

We shall also use
\begin{equation}\label{eq:conditional-compose}
 K(w,y\mid x)\le K(w\mid x)+K(y\mid w)+C_{\rm comp}.
\end{equation}
To obtain~\eqref{eq:conditional-compose}, apply~\eqref{eq:precise-chain} to get
\[
 K(w,y\mid x)=K(w\mid x)+K(y\mid x,w,K(w\mid x))+O(1).
\]
A program that outputs $y$ given $w$ may ignore the additional conditions $x$ and $K(w\mid x)$, so
\[
 K(y\mid x,w,K(w\mid x))\le K(y\mid w)+O(1).
\]
Combining the two fixed error terms gives~\eqref{eq:conditional-compose}, with $C_{\rm comp}$ independent of $x,w,y$.

\subsection{Uniform interval counting}
For $x,y\in\bits^*$ and $\delta\ge0$, define
\begin{equation}\label{eq:interval-def}
 I_\delta(x,y)=\{w\in\bits^*: \dK(x,w)+\dK(w,y)\le\dK(x,y)+\delta\}.
\end{equation}
This is the rough interval, or $\delta$-interval, of the metric $\dK$; see~\cite[Definition 1.1, p.~2926]{chatterji}.

A related counting argument appears in Romashchenko~\cite[Lemma 1, p.~5]{romashchenko}: many intermediate strings satisfying two prescribed bounds on conditional program lengths yield a shorter description between the endpoints. That result uses plain complexity and a logarithmic remainder depending on these bounds. Here the conditional prefix chain rule gives a bound depending only on the excess $\delta$ and a fixed constant.

\begin{lemma}\label{lem:interval}
There is a constant $C\ge0$ such that, for all $x,y\in\bits^*$ and all real $\delta\ge0$,
\begin{equation}\label{eq:interval-card}
 |I_\delta(x,y)|\le 2^{\delta+C}.
\end{equation}
In particular, every such interval is finite.
\end{lemma}
\begin{proof}
A fixed identity program gives $K(x\mid x)\le C_{\rm id}$ uniformly. Together with~\eqref{eq:distance}, this gives
\begin{equation}\label{eq:direction-distance}
 K(x\mid y)\le\dK(x,y)+C_{\rm id}
\end{equation}
also on the diagonal. Suppose first that $x\ne y$. Since $I_\delta(x,y)=I_\delta(y,x)$, interchange $x$ and $y$, if necessary, so that $k:=K(y\mid x)\ge K(x\mid y)$. Then $\dK(x,y)=k+c_{\rm ref}$. For every $w\in I_\delta(x,y)$, equations~\eqref{eq:conditional-compose} and~\eqref{eq:direction-distance} yield, with a uniform nonnegative $C_1$,
\[
 K(w,y\mid x)\le\dK(x,w)+\dK(w,y)+O(1)
 \le k+\delta+C_1.
\]
The lower direction of~\eqref{eq:precise-chain}, with the fixed cost of exchanging the pair included, gives a uniform nonnegative $C_2$ such that
\[
 k+K(w\mid x,y,k)\le K(w,y\mid x)+C_2.
\]
Consequently
\begin{equation}\label{eq:interval-description}
 K(w\mid x,y,k)\le\delta+C_3,\qquad C_3=C_1+C_2.
\end{equation}
Once the strings $x,y$ are fixed, $k=K(y\mid x)$ is also fixed. Thus, the programs producing the strings $w\in I_\delta(x,y)$ all receive the same condition $(x,y,k)$. There are $2^{N+1}-1$ binary strings of length at most $N=\lfloor\delta+C_3\rfloor$, each of which can describe at most one output under this condition. Hence
\[
 |I_\delta(x,y)|\le2^{N+1}-1
 =2^{\lfloor\delta+C_3\rfloor+1}-1
 \le2^{\delta+C_3+1}.
\]
Taking $C\ge C_3+1$ gives $|I_\delta(x,y)|\le2^{\delta+C}$ when $x\ne y$.

If $x=y$, membership in the interval implies $2\dK(x,w)\le\delta$. By~\eqref{eq:direction-distance}, $K(w\mid x)\le\delta/2+C_{\rm id}$. The same program count gives $|I_\delta(x,x)|\le2^{\delta/2+C_{\rm id}+1}\le2^{\delta+C_{\rm id}+1}$. Taking $C=\max\{C_3+1,C_{\rm id}+1\}$ gives~\eqref{eq:interval-card} in both cases.
\end{proof}

\subsection{From the integer line to integer lattices}
To rule out an embedding of the integer line $\mathbb Z$, it suffices to consider its nonnegative part $\mathbb N_0=\{0,1,2,\ldots\}$. Any map on $\mathbb Z$ satisfying the required distance estimate would restrict to a map on $\mathbb N_0$ with the same scale and uniform error. The following theorem rules out such a map already at any fixed positive scale.

\begin{theorem}\label{thm:noray}
There do not exist a map $f:\mathbb N_0\to\bits^*$, a real $s>0$ and a finite real $E\ge0$ such that
\begin{equation}\label{eq:ray-error}
 |\dK(f(i),f(j))-s|i-j||\le E
 \qquad(i,j\in\mathbb N_0).
\end{equation}
Neither computability nor injectivity of $f$ is assumed.
\end{theorem}
\begin{proof}
Put $x_i=f(i)$ and fix any integer $n\ge1$. For $0\le i\le n$,
\[
 \dK(x_0,x_i)+\dK(x_i,x_n)\le si+E+s(n-i)+E=sn+2E,
\]
whereas $\dK(x_0,x_n)\ge sn-E$. Thus
\[
 \{x_0,\ldots,x_n\}\subseteq I_{3E}(x_0,x_n),
 \qquad |\{x_0,\ldots,x_n\}|\le2^{3E+C}.
\]
If $x_i=x_j$, then~\eqref{eq:ray-error} gives $s|i-j|\le E$. The indices corresponding to any one string lie in an integer interval of length at most $\lfloor E/s\rfloor$, starting at the least such index. Each string therefore occurs at most $R:=\lfloor E/s\rfloor+1$ times. Counting the $n+1$ indices by their images gives
\begin{equation}\label{eq:finite-chain-bound}
 n+1\le \left(\left\lfloor E/s\right\rfloor+1\right)2^{3E+C}.
\end{equation}
The right side is fixed independently of $n$, which is impossible for arbitrarily large $n$.
\end{proof}

\begin{proof}[Proof of Corollary~\ref{cor:lattice}]
Suppose there is a string scale-embedding $\Phi_s:\mathbb Z^m\to\bits^*$. Fix a positive integer scale $s$ at which its error estimate holds. At this scale there is a finite $E(s)\ge0$, independent of the lattice points, such that
\[
 \left|\dK(\Phi_s(a),\Phi_s(b))-s\|a-b\|_1\right|\le E(s)
 \qquad(a,b\in\mathbb Z^m).
\]
Define a new map on the nonnegative integers by
\[
 \psi_s:\mathbb N_0\to\bits^*,\qquad
 \psi_s(n)=\Phi_s((n,0,\ldots,0)).
\]
For all $i,j\in\mathbb N_0$,
\[
 \|(i,0,\ldots,0)-(j,0,\ldots,0)\|_1=|i-j|,
\]
and substituting these lattice points into the assumed estimate gives
\[
 \left|\dK(\psi_s(i),\psi_s(j))-s|i-j|\right|\le E(s).
\]
Thus the assumed lattice embedding would give a map of the integer ray with the same positive scale and finite uniform error. Theorem~\ref{thm:noray} rules out precisely such a map, giving a contradiction. No map can satisfy the required distance estimate.

A sequence scale-embedding gives the same contradiction. At the fixed scale $s$, take the selected finite prefix at each lattice point to obtain $\Phi_s$, then define $\psi_s$ as above. The prefix lengths may depend on the lattice point; this does not affect the distance estimate or its restriction.
\end{proof}

\begin{remark}\label{rem:lattice-scope}
The use of prefix $K$ is essential to the present argument: the logarithmic comparison errors with the optional monotone complexity $K_m$ in~\cite[Assumption 10]{hutter} need not be uniform over an unbounded source. The interval proof uses no triangle inequality and holds for every fixed $c_{\rm ref}\ge0$.
\end{remark}

\begin{corollary}\label{cor:realspace}
For every integer $m\ge1$, $(\mathbb R^m,\|\cdot\|_1)$ admits neither a string scale-embedding nor a sequence scale-embedding into $B_K^*$.
\end{corollary}
\begin{proof}
The coordinatewise inclusion $\mathbb Z^m\subseteq\mathbb R^m$ preserves $\ell_1$ distances. Restricting any proposed string embedding of $\mathbb R^m$ to this subset would therefore give one of $\mathbb Z^m$, with the same scales and uniform error constants, contrary to Corollary~\ref{cor:lattice}. For a sequence embedding, restrict the assigned sequences and their prefix lengths to the integer points; all distance and length conditions are preserved.
\end{proof}
This corollary shows that the boundedness hypothesis in~\cite[Corollary 24]{hutter} cannot be dropped for the whole real $\ell_1$ space. It is a consequence of the lattice obstruction and does not assert non-embeddability for every unbounded subset.

\section{Concluding remarks}\label{sec:conclusion}
Our finite-field construction gives scale-embeddings of every finite complete bipartite graph into $B_K^*$, with constant error for finite strings and logarithmic error for prefixes of fixed infinite sequences. Combined with Hutter's Theorem 36~\cite{hutter}, these embeddings show that no positive power of information distance can be represented exactly by distances in a real Hilbert space. The interval bound also rules out scale-embeddings of the whole integer line into $B_K^*$: even at a single positive scale, no finite error bound can hold for all integer pairs.

\section*{Acknowledgements}
This work was developed through an iterative process of discussion and revision between the author and ChatGPT. The author developed the proof ideas through these discussions, and ChatGPT wrote the mathematical proofs based on those ideas. The author carefully reviewed the mathematical proofs to verify their correctness and ensure that they accurately reflected the author's intended arguments.

The accompanying Lean formalization was written by ChatGPT. The author checked that the hypotheses and conclusions of the main formalized theorems correspond to those in the paper. The formal proofs were checked by the Lean kernel.

The author takes full responsibility for the paper and the accompanying code.

\appendix
\section{Scope of formal verification}
The main results are formalized in Lean 4.33.1, using mathlib 4.33.1 and Alexey Milovanov's library for conditional Kolmogorov complexity. The code is provided as supplementary material. It uses library results for standard complexity estimates and the Hilbert-space identity; the matrix construction, coding bounds and non-embedding arguments are proved in the accompanying code. The complexity library is fixed to commit \texttt{97c6dfc} of \url{https://github.com/AlexeyMilovanov/kolmogorov-complexity-lean}.

The theorem \texttt{library\_bipartite\_embedding\_main} establishes both parts of Theorem~\ref{thm:main}, using the library's definition of optimal conditional prefix complexity on binary strings represented as \texttt{List Bool}. It includes the same positive parameters $r,B$ for both parts, the exact lengths~\eqref{eq:length}, injectivity, both error bounds, and the monotonicity and linear growth of the selected prefix lengths. The formal statement also covers $s=0$; restricting it to $s\ge1$ gives the result stated in the paper.

The Hilbert-space conclusions and the identification of graph distance are formalized as follows:
\begin{itemize}
\item \texttt{library\_information\_distance\_nonembedding} proves Corollary~\ref{cor:allpowers};
\item \texttt{library\_infinite\_prefix\_nonembedding} proves Corollary~\ref{cor:prefix}, including $N>1$, fixed infinite binary sequences and the exact common prefix lengths;
\item \texttt{graphDistance\_eq\_shortestPath} identifies the explicit distance taking values $0,1,2$ with mathlib's shortest-path distance on the complete bipartite graph, assuming both sides are nonempty.
\end{itemize}
The infinite-source construction follows Section~\ref{sec:infinite}. It uses the library's existence theorem for Martin--L\"of random sequences and its Levin--Schnorr theorem, then establishes the logarithmic cost of conditioning on the block count. The theorem \texttt{nested\_source\_martinLof} shows that the packed matrix blocks are prefixes of the chosen random sequence, and \texttt{martinLof\_source\_not\_computable} derives its noncomputability from the complexity lower bound. As in Section~\ref{sec:powers}, the proof of Corollary~\ref{cor:allpowers} uses the finite-string embeddings and their constant error.

For the interval argument, \texttt{library\_\allowbreak information\_\allowbreak intervals\_\allowbreak real\_\allowbreak bound} proves Lemma~\ref{lem:interval} for the same \texttt{libraryInformationDistance} used in the embedding theorems. It applies the exact conditional chain rule with the complexity value retained in the auxiliary input. The bound on the number of strings with bounded conditional prefix complexity comes from the library. The composition estimate follows from the upper bound in the conditional chain rule and the option to ignore extra conditions, with a fixed cost for exchanging the components of an encoded pair.

The subsequent ray argument includes the bound on repeated images. The string and sequence conclusions of Corollary~\ref{cor:lattice} are proved by \texttt{library\_\allowbreak no\_\allowbreak lattice\_\allowbreak string\_\allowbreak scale\_\allowbreak embedding} and \texttt{library\_\allowbreak no\_\allowbreak lattice\_\allowbreak sequence\_\allowbreak scale\_\allowbreak embedding}. These theorems impose neither computability nor injectivity assumptions on the proposed embeddings, and allow the error estimate to hold only at sufficiently large positive integer scales.

The corresponding results for Corollary~\ref{cor:realspace} are \texttt{library\_\allowbreak no\_\allowbreak real\_\allowbreak string\_\allowbreak scale\_\allowbreak embedding} and \texttt{library\_\allowbreak no\_\allowbreak real\_\allowbreak sequence\_\allowbreak scale\_\allowbreak embedding}. The real $\ell_1$ distance is defined as the finite sum of absolute coordinate differences. Both proofs restrict a proposed embedding to integer-coordinate points and apply the lattice results, without any additional complexity estimate.

Some encoding choices differ from those in the paper. The formalization uses fixed finite bijections to encode field elements, pairs and matrices in $r$, $2r$ and $4r$ bits, respectively, in place of basis-coordinate encodings. For each fixed choice of encodings, the corresponding decoders and block transformations are computable. Their input and output lengths match those used in the proofs. Vertex labels are assigned by an injection into binary strings of a fixed length. A length-preserving conversion relates the internal string representation to the library's \texttt{List Bool} representation, and the corresponding conditional complexities are proved equal. The existence of an optimal machine is supplied by the library.

The nonrepresentation theorems hold for every nonnegative reference constant and do not use the triangle inequality. The fact that a sufficiently large reference constant makes~\eqref{eq:distance} a metric is taken from Hutter and is not formalized here. The formalization uses prefix $K$ throughout; monotone $K_m$ is outside its scope.


\begin{thebibliography}{9}
\bibitem{bennett}
Bennett CH, G\'acs P, Li M, et al.
Information distance.
\emph{IEEE Trans Inf Theory} 1998; \textbf{44}(4): 1407--1423.
\href{https://doi.org/10.1109/18.681318}{doi:10.1109/18.681318}.

\bibitem{hutter}
Hutter M.
Properties of algorithmic information distance.
\emph{IEEE Trans Inf Theory} 2025; \textbf{71}(10): 7540--7554.
\href{https://doi.org/10.1109/TIT.2025.3597092}{doi:10.1109/TIT.2025.3597092}.
Definition, assumption, result and problem numbers in the present article refer to arXiv:2507.21988v1,
\url{https://arxiv.org/abs/2507.21988v1}.

\bibitem{shen}
Shen A, Uspensky VA and Vereshchagin N.
\emph{Kolmogorov Complexity and Algorithmic Randomness}.
Providence, RI: American Mathematical Society, 2017.
Section, theorem and problem numbers and page references refer to the authors' 2016 draft,
\url{https://www.lirmm.fr/~ashen/kolmbook-eng.pdf}.

\bibitem{schoenberg}
Schoenberg IJ.
Metric spaces and positive definite functions.
\emph{Trans Amer Math Soc} 1938; \textbf{44}(3): 522--536.
\href{https://doi.org/10.1090/S0002-9947-1938-1501980-0}{doi:10.1090/S0002-9947-1938-1501980-0}.

\bibitem{chatterji}
Chatterji I and Dru\c{t}u C.
Median geometry for spaces with measured walls and for groups.
\emph{Math Ann} 2025; \textbf{393}: 2925--2952.
\href{https://doi.org/10.1007/s00208-025-03289-1}{doi:10.1007/s00208-025-03289-1}.

\bibitem{romashchenko}
Romashchenko A.
Clustering with respect to the information distance.
\emph{Theor Comput Sci} 2022; \textbf{929}: 164--171.
\href{https://doi.org/10.1016/j.tcs.2022.06.039}{doi:10.1016/j.tcs.2022.06.039}.
The lemma number and page reference refer to arXiv:2110.01346v4,
\url{https://arxiv.org/abs/2110.01346v4}.
\end{thebibliography}
\end{document}